\documentclass[%
 reprint,
 amsmath,amssymb,
 aps,
]{revtex4-1}

\usepackage{mathrsfs}

\usepackage{graphicx}
\usepackage{stmaryrd}
\usepackage{amssymb}
\usepackage{amsmath, cancel, centernot }
\usepackage{amsthm}
\usepackage[mathscr]{eucal}
\usepackage{mathtools}
\usepackage{epstopdf}

\usepackage{graphicx}
\usepackage{dcolumn}
\usepackage{bm}

\makeatletter
\newcommand\be{\@ifstar{\[}{\begin{equation}}}
\newcommand\ee{\@ifstar{\]}{\end{equation}}}
\newcommand\bp{\begin{pmatrix}}
\newcommand\ep{\end{pmatrix}}

\makeatother

\begin{document}


\title{Experimental Non-Violation of the Bell Inequality}

\author{T.N.Palmer}
 \email{tim.palmer@physics.ox.ac.uk}
\affiliation{%
 Department of Physics\\
 University of Oxford
}%


\date{\today}

\begin{abstract}

A finite non-classical framework for physical theory is described which challenges the conclusion that the Bell Inequality has been shown to have been violated experimentally, even approximately. This framework postulates the universe as a deterministic locally causal system evolving on a measure-zero fractal-like geometry $I_U$ in cosmological state space. Consistent with the assumed primacy of $I_U$, and $p$-adic number theory, a non-Euclidean (and hence non-classical) metric $g_p$ is defined on cosmological state space, where $p$ is a large but finite Pythagorean prime. Using number-theoretic properties of spherical triangles, the inequalities violated experimentally are shown to be $g_p$-distant from the CHSH inequality, whose violation would rule out local realism. This result fails in the singular limit $p=\infty$, at which $g_p$ is Euclidean. Broader implications are discussed. 
\end{abstract}

\pacs{Valid PACS appear here}
\maketitle


\section{\label{sec:level1}Introduction}

Recent experiments (e.g. \cite{Shalm}) have seemingly put beyond doubt the conclusion that the CHSH version 
\begin{equation}
\label{CHSH}
|\text{Corr}(0,0)+\text{Corr}(1,0)+\text{Corr}(0,1)-\text{Corr}(1,1)| \le 2
\end{equation}
of the Bell Inequality is violated robustly for a range of experimental protocols and measurement settings. As a result, it is widely believed that physical theory cannot be based on Einsteinian notions of realism and local causality (`local realism'). Here, $\text{Corr}(X,Y)$ denotes the correlation between spin measurements performed by Alice and Bob on entangled particle pairs produced in the singlet quantum state, where $X=0, 1$ and $Y=0,1$ correspond to pairs of freely-chosen points on Alice and Bob's celestial spheres, respectively.  

Of course, in the \emph{precise} form as written, (\ref{CHSH}) has not been shown to have been violated experimentally. In practice, the four correlations on the left-hand side of (\ref{CHSH}) are each estimated from a separate sub-ensemble of particles with measurements performed at different times and/or spatial locations. Hence, for example, the measurement orientation corresponding to $Y=0$ for the first sub-ensemble cannot correspond to \emph{precisely} the same measurement orientation $Y=0$ for the second sub-ensemble; as a matter of principle Bob cannot shield his apparatus from the effects of gravitational waves associated for example with distant astrophysical events. Hence, as a matter of principle, what is actually violated experimentally is not (\ref{CHSH}) but 
\be
\label{CHSHmod}
|\text{Corr}(0,0)+\text{Corr}(1,0')+\text{Corr}(0',1)-\text{Corr}(1',1')| \le 2
\ee
where, relative to the Euclidean metric of space-time, $0\approx 0'$ and $1 \approx 1' $ for $X$ and $Y$.

Could the difference between  $0\approx 0' $, $1 \approx 1' $ on the one hand, and $0 = 0'$, $1 = 1'$ on the other, actually matter? More specifically, is there a plausible framework for physical theory where (\ref{CHSH}) is the singular \cite{Berry} rather than the smooth limit of (\ref{CHSHmod}) as $0' \rightarrow 0$, $1' \rightarrow 1$ and where (\ref{CHSHmod}) is therefore distant from (\ref{CHSH}) no matter how accurate are our finite-precision experiments? Intuitively it would seem not, as Bell \cite{Bell:1964} himself argued with a form of `epsilonic' analysis. However, this and related analyses \cite{WoodSpekkens} crucially assume that the metric of \emph{state space} (in contrast with space-time) is Euclidean. In conventional theory (classical and quantum) where state space is constructed from the fields $\mathbb R$ and $\mathbb C$ and where the Euclidean metric plays a vital role in defining $\mathbb R$ and $\mathbb C$ as continuations of the more primitive rationals $\mathbb Q$, this assumption is implicit. However, motivated by both nonlinear dynamical systems theory and $p$-adic number theory, we outline in Section \ref{IST} a plausible and robust locally causal framework where the metric on state space is explicitly not Euclidean. Within this framework, complex Hilbert states can have probabilistic interpretations as uncertain trajectory segments on a measure-zero fractal invariant set $I_U$ in cosmological state space, providing certain descriptors of these states take rational values. By implication, states with irrationally-valued descriptors have no ontic status as elements of $I_U$. A metric $g_p$ (where $p$ is a large prime) is introduced on state space which respects the primacy of $I_U$ and with respect to which such ontic and non-ontic states are necessarily distant from one another. In Section \ref{CHSHsec} it is shown that the violation of (\ref{CHSHmod}) is generically robust to $g_p$-small-amplitude perturbations. However, the set of all inequalities encompassed by such perturbations does not and cannot include the Bell inequality (\ref{CHSH}), whose violation is needed to rule out local realism (\ref{CHSH}). As shown, (\ref{CHSH}) is necessarily constructed from states with irrational descriptors, i.e. non-ontic states not lying on $I_U$. In this sense (\ref{CHSH}) is $g_p$-distant from (\ref{CHSHmod}). 

Using $g_p$, Meyer's \cite{Meyer:1999} result that finite-precision measurements can negate quantum no-go theorems has been extended to the Bell Theorem. In particular,  experiment has not ruled out, even approximately, realistic, locally causal theories of physics based on non-Euclidean state-space metrics. Some implications for theoretical physics more generally are discussed in Section \ref{conclusions}. 
\bigskip

\section{Invariant Set Theory}
\label{IST}

Results below summarise more detailed analysis given in \cite{Palmer:2017}. We treat the universe $U$ as a self-contained locally causal deterministic system evolving over multiple epochs on some compact measure-zero fractal invariant set $I_U$ in cosmological state space \cite{Palmer:2009a} \cite{Palmer:2014}. Fig \ref{fractal} illustrates the state-space geometry of $I_U$ locally. On the left is shown, at some $j$th fractal iterate of $I_U$, a single state-space trajectory segment (`history') associated with some isolated sub-system evolving over time. At the $j+1$th iterate, this trajectory segment comprises a helix of $N \gg 0$ fine-scale trajectories and an additional $N+1$th trajectory (not shown) at the centre of the helix. Here $1/N$ is plausibly linked to the gravitational coupling constant, the squared mass of the electron in Planck units i.e. $O(10^{-45})$, but results here do not depend on the specific value of $N$ (except that, to link with the complex Hilbert vector formalism below, $N$ must be divisible by 4). At higher iterates (not shown), elements of this helix are themselves helical.  Hence, locally, $I_U$ is a Cantor set $\hat C(p)$ of trajectory segments, i.e. equal to the product $\hat C(p) \times \mathbb R$, where $p=N+1$, and $\hat C(p)$ is based on $p$ iterated elements based on a regular $N$-gon, with an additional  $N+1$th element at the centre \cite{Katok}. Points in state space which do not lie on $I_U$ (e.g. lie in the gaps in the helix) have no ontic status and do not correspond to elements of physical reality.  It is assumed that the laws of physics in their most primitive form derive from the geometry of $I_U$. Here, for example, the relationship $E=\hbar \omega$, relates a source term for the geometry of space-time to the slope $\omega$ of the $j+1$th iterate trajectories as they wind around a $j$th iterate trajectory in state space. In practice, it can be assumed that these fractal iterates terminate at some finite $J$th value and that $I_U$ is a fractal-like limit cycle.

The $N$ trajectory segments in Fig \ref{fractal} are labelled according to a process illustrated in Fig \ref{decohere}a, associated with the divergence and nonlinear clustering of trajectories into two distinct state-space clusters labelled $a$ and $\cancel a$. This corresponds to the generic phenomenon of decoherence as the sub-system interacts with its environment. Hence, each fine-scale trajectory in the helix in Fig \ref{fractal} can be labelled $a$ or $\cancel a$ according to whether the trajectory evolves to the $a$ or the $\cancel a$ cluster. The $N+1$th trajectory at the centre of the helix is presumed to evolve to the unstable equilibrium between clusters. A second period of decoherence is also shown involving a helix of $j+2$nd iterate trajectory segments. At the $j$th iterate (Fig \ref{decohere}b), such decoherence appears to resemble the branching process of the Everettian interpretation. However, this is illusory and moreover neighbouring trajectories do not correspond to `many worlds' but to the (mono-) universe at later or earlier cosmological epochs. In the analysis below, $j$ plausibly terminates finitely, making $I_U$ a finite fractal-like limit cycle. 

\begin{figure}
\centering
\includegraphics[scale=0.2]{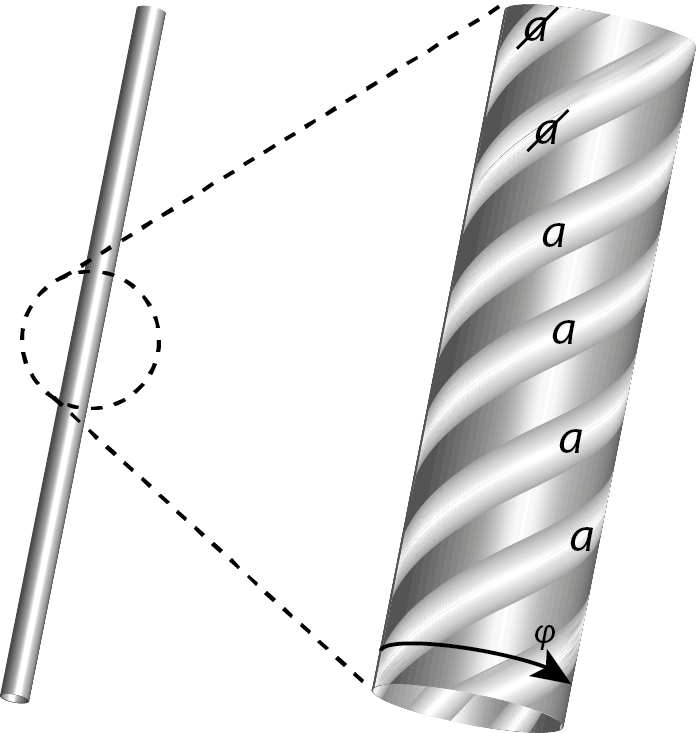}
\caption{\emph{A trajectory segment at some $j$th fractal iterate is found to comprise a helix of $N \gg 0$  $j+1$th iterate trajectory segments. The fine-scale trajectory segments are each labelled symbolically $a$ or $\cancel a$ according to the cluster to which they evolve under the process of decoherence (illustrated in Fig \ref{decohere} below). The segments all belong to a single trajectory corresponding to a compact measure-zero fractal invariant set $I_U$ on which the universe evolves in cosmological state space.}}
\label{fractal}
\end{figure}

Within this geometric framework, `reality' can be considered some uncertain $1\le I \le N$th fine-scale trajectory segment, and, as discussed in \cite{Palmer:2017}, can be represented probabilistically by the complex Hilbert vector
\be
\label{complexqubit}
\cos \frac{\theta}{2} \; |a\rangle+  e^{i \phi} \sin \frac{\theta}{2} \;|\cancel a \rangle
\ee
where $\cos^2 \theta/2$ equals the fraction $n_1/N$ of trajectories labelled $a$, and $\phi$ describes a fractional rotation $2\pi n_2/N$ of the helix ($1\le n_1 \le N,\  1 \le n_2 \le N$). The orthogonal eigenvectors $|a\rangle$ and $|\cancel a \rangle$ correspond to the two distinct clusters $a$ and $\cancel a$. In particular it is necessary that $\phi/2\pi$, $\cos^2 \theta/2$, and hence $\cos \theta$ are rational numbers. By contrast, a putative Hilbert vector where $\cos \theta \notin \mathbb Q$ or $\phi/2\pi \notin \mathbb Q$ cannot correspond (probabilistically) to any trajectory segment on $I_U$ and therefore cannot correspond to an ontic state. More general tensor-product Hilbert states can also be used to represent uncertain multi-variate properties of the $I$th trajectory segment. Again it is necessary that all squared amplitudes are rational, and all complex phase angles are rational multiples of $2\pi$. The following number-theoretic result is key:

\begin{figure}
\centering
\includegraphics[scale=0.3]{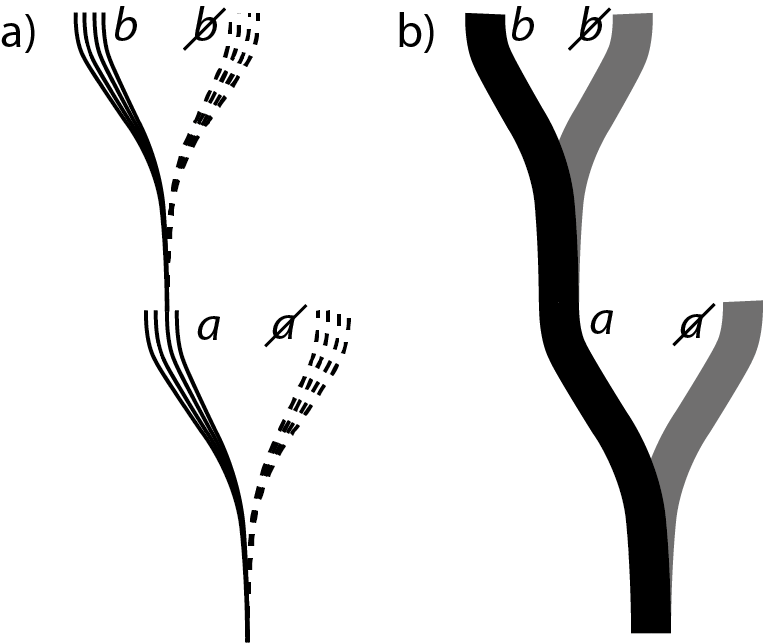}
\caption{\emph{a) Associated with environmental decoherence, $j+1$th iterate trajectories are shown diverging and clustering into two clusters $a$ and $\cancel a$. Similarly, at later time, $j+2$th iterate trajectories diverge and cluster into two clusters $b$ and $\cancel b$. In this way, time can be parametrised by fractal iterate number. b) From the perspective of the $j$th iterate of $I_U$, a) resembles Everettian `branching'.}}
\label{decohere}
\end{figure}

\newtheorem{theorem}{Theorem}
\begin{theorem}
\label{theorem}
 Let $\phi/\pi \in \mathbb{Q}$. Then $\cos \phi \notin \mathbb{Q}$ except when $\cos \phi =0, \pm \frac{1}{2}, \pm 1$. \cite{Niven, Jahnel:2005}
\end{theorem}
\begin{proof} 
Assume that $2\cos \phi = a/b$ where $a, b \in \mathbb{Z}, b \ne 0$ have no common factors.  Since $2\cos 2\phi = (2 \cos \phi)^2-2$, then $2\cos 2\phi = (a^2-2b^2)/b^2$. Now $a^2-2b^2$ and $b^2$ have no common factors, since if $p$ were a prime number dividing both, then $p|b^2 \implies p|b$ and $p|(a^2-2b^2) \implies p|a$, a contradiction. Hence if $b \ne \pm1$, then the denominators in $2 \cos \phi, 2 \cos 2\phi, 2 \cos 4\phi, 2 \cos 8\phi \dots$ get bigger without limit. On the other hand, if $\phi/\pi=m/n$ where $m, n \in \mathbb{Z}$ have no common factors, then the sequence $(2\cos 2^k \phi)_{k \in \mathbb{N}}$ admits at most $n$ values. Hence we have a contradiction. Hence $b=\pm 1$ and $\cos \phi =0, \pm\frac{1}{2}, \pm1$. 
\end{proof}
We now define a metric $g_p$ where ontic states on $I_U$ and non-ontic states off $I_U$ are necessarily distant from one another (no matter how close they may appear from a Euclidean perspective). By Ostrowsky's theorem, the Euclidean and $p$-adic metrics are the only inequivalent norm-induced metrics \cite{Katok} \cite{Khrennikov}. For algebraic reasons it is important for $p$ to be prime. Hence, with $N$ divisible by 4, $p=N+1 \gg 1$ must be a Pythagorean prime, the sum of two squared integers. Let $x$ and $y$ denote trajectory segments on $I_U$, i.e. points on $\hat C(p)$. Firstly if both $x, y \in \hat C(p)$ then $g_p \le 1$ is Euclidean. This is consistent with the fact that the set of $p$-adic integers $\mathbb Z_p$ is homeomorphic to $\hat C(p)$ \cite{Katok}. Secondly, if at least one of $x, y \notin \hat C(p)$ then $g_p(x,y)=p$ if $x \ne y$, otherwise $g_p(x,y)=0$. This is consistent with the notion that the $p$-adic distance between a $p$-adic integer and a non-integer $p$-adic number is at least $p$ and therefore $\gg 1$. It is easily shown that $g_p$ satisfies the axioms for a metric (e.g. the triangle inequality) on cosmological state space. 

\section{The Bell Inequality}
\label{CHSHsec}

Consider now the relationship between (\ref{CHSH}) and (\ref{CHSHmod}) from the perspective of the framework above. As mentioned, $X=0,1$, $Y=0,1$ are four random points on the sphere, three of which (relevant to the discussion below) are shown in Fig \ref{F:CHSH}a. Let $\theta_{XY}$ denote the relative orientation between an $X$ point and a $Y$ point. Recall that complex Hilbert states can represent the multi-variate probabilistic elements of trajectory segments on $I_U$ providing squared amplitudes are rational. Hence,  in invariant set theory, $\text{Corr}(X,Y)=-\cos \theta_{XY}$ providing $\cos \theta_{XY} \in \mathbb Q$. 

Let us suppose that Alice freely chooses $X=0$ and Bob $Y=0$ when measuring a particular entangled particle pair. Then it must be the case that $\cos \theta_{00} \in \mathbb Q$.  Could Alice and Bob have chosen $X=1$ and $Y=0$, given that they actually chose $X=0$, $Y=0$ respectively? In other words, does the world in which this counterfactual experiment takes place also lie on $I_U$? To answer the question in the affirmative we require $\cos \theta_{10} \in \mathbb Q$. However, from the cosine rule for spherical triangles, we have 
\be 
\label{cosinerule}
\cos \theta_{10}=\cos \theta_{00} \cos \alpha_X + \sin \theta_{00} \sin \alpha_X \cos \gamma
\ee
where $\alpha_X$ is the angular distance between $X=0$ and $X=1$. Now it is always possible for Alice to send the particle which she has just measured in the $X=0$ direction, back into the measuring apparatus to be again measured in the $X=1$ direction. Hence $\cos \alpha_X$ must be rational. Now, we also require the angle $\gamma$ to be a rational multiples of $2 \pi$. This would be so if the three directions $X=0,1$ and $Y=0$ were coplanar \emph{exactly}, so that $\gamma=180^\circ$ \emph{precisely}. However, because of ubiquitous unshieldable gravitational waves, this cannot be the case. Hence, $\cos \theta_{01}$ is the sum of two terms, the first a rational and the second the product of three independent terms, the last of which is irrational. Being independent, these three terms cannot conspire to make their product rational. Hence $\cos \theta_{01}$ is the sum of a rational and an irrational and must therefore be irrational. Hence the world in which the counterfactual experiment takes place has no ontic status and is $g_p$-distant from worlds on $I_U$. Hence the counterfactual question cannot be answered in the affirmative: $\text{Corr}(1,0)$ is undefined. In general it is never the case that all four correlations in (\ref{CHSH}) are definable on $I_U$ - the Bell inequality is always undefined in invariant set theory.

\begin{figure}
\centering
\includegraphics[scale=0.3]{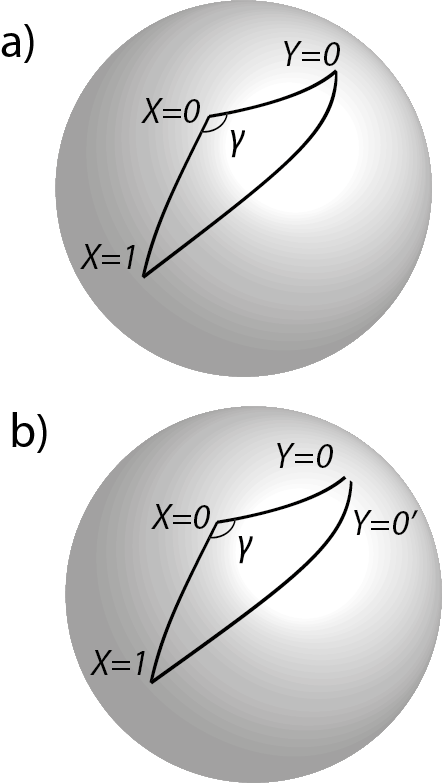}
\caption{\emph{a) In general it is impossible for all the cosines of the angular lengths of all three sides of the spherical triangle to be rational, and the internal angles rational multiples of $2\pi$. b) What actually occurs when (\ref{CHSHmod}) is tested experimentally. Here the cosines of the angular lengths of all sides are rational. In a precise sense b) is $g_p$ distant from a).}}
\label{F:CHSH}
\end{figure}

An experimenter might ask how one could set up an experiment with sufficient care to ensure that the corresponding Hilbert state descriptors were rational rather than irrational. The experimenter need take no care: if an experiment is performable, i.e. corresponds to some $U \in I_U$, then by construction the descriptors must be rational. Physical perturbations (e.g. gravitational waves in space-time) only introduce uncertainty in the values of the rational descriptors and not in the fact that they are rational. Conversely, if the descriptor of a counterfactual state is irrational, then no amount of noise which respects the primacy of $I_U$ can change it into an ontic state. This property provides an attractive finitist feature which is missing in conventional physical theories based on $\mathbb R$ or $\mathbb C$ and hence on the Euclidean state-space metric. Hence in the real world of experiment, both $\cos \theta_{00}$ and $\cos\theta_{10'}$ in (\ref{CHSHmod}) are necessarily and robustly rational (Fig \ref{F:CHSH}b), consistent with the fact the individual sub-ensembles are measured at different times, and that unshieldable gravitational waves ensure that orientations are not \emph{precisely} the same when these different sub-ensembles are measured. Indeed we can infer the existence of an effectively infinite family of orientations where all of $\cos \theta_{00}, \cos\theta_{10'}, \cos \theta_{0'1}, \cos\theta_{1'1'}$ in (\ref{CHSHmod}) are rational. However, by construction, none of the orientations so generated includes those associated with (\ref{CHSH}), which is therefore indeed the singular limit of and $g_p$-distant from (\ref{CHSHmod}). Just as the paradox of the Penrose `impossible triangle' is resolved by realising that the sides of the triangle are not necessarily close near a vertex of the triangle, so too here. As discussed in \cite{Palmer:2017}, many of the familiar `paradoxes' of quantum theory can be interpreted realistically and causally with $g_p$ as the metric of state space. 

\section{Conclusions}
\label{conclusions}

A theoretical framework based on invariant set theory has been outlined, which asserts that no physical experiment can or will be able to demonstrate that the Bell inequality (\ref{CHSH}) is violated - even approximately. In this framework, (\ref{CHSH}) is the singular limit of (\ref{CHSHmod}), and in the limit, (\ref{CHSH}) is neither satisfied nor violated: it is undefined. Key to this formulation, and motivated by $p$-adic number theory, a non-Euclidean metric $g_p$ is introduced on state space, where $p$ is a large but finite Pythagorean prime. $g_p$ respects the primacy of an assumed fractal geometry on which the universe evolves and from which the laws of physics derive. The Euclidean metric (e.g. of classical theory \cite{Richens}) arises in the singular limit $p=\infty$, whence fractal gaps shrink to zero. In this singular limit, the conclusions above fail - whence the experimental violation of (\ref{CHSHmod}) does rule out conventional classical theories of quantum physics. Being realistic, geometric and locally causal, the analysis above suggests a novel holistic perspective on the most important open question in theoretical physics today: how to synthesise quantum and gravitational physics. Given the unshieldable effects of gravitational waves in the discussion above (and other reasons discussed in \cite{Palmer:2017}), it is proposed that gravitational processes play a vital role, not only during decoherent measurement \cite{diosi:1989} \cite{Penrose:2004}, but also during non-decoherent unitary evolution. As such, this long-sought synthesis may more likely arise from a locally realistic but non-classical `gravitational theory of the quantum' than from a nonlocal or nonrealistic `quantum theory of gravity' \cite{Palmer:2017}.

\bibliographystyle{plain}
\bibliography{mybibliography}

\end{document}